\documentclass[letterpaper, 10 pt, conference]{ieeeconf}  

\IEEEoverridecommandlockouts                              
\usepackage{amsmath,amssymb,amsfonts}
\usepackage{mathtools}
\usepackage{mathrsfs}
\usepackage{bm}

\usepackage{algorithmic}
\usepackage{graphicx} 
\usepackage{color} 
\usepackage{textcomp}

\newcommand{\BA}{\bm{A}}
\newcommand{\BB}{\bm{B}}
\newcommand{\BJ}{\bm{J}}
\newcommand{\BW}{\bm{W}}
\newcommand{\BX}{\bm{X}}

\newcommand{\Bc}{\bm{c}}
\newcommand{\Bd}{\bm{d}}
\newcommand{\Bf}{\bm{f}}
\newcommand{\Bn}{\bm{n}}

\newcommand{\Bx}{\bm{x}}
\newcommand{\BLambda}{\boldsymbol{\Lambda}}
\newcommand{\Bmu}{\boldsymbol{\mu}}
\newcommand{\calF}{\mathcal{F}}
\newcommand{\calJ}{\mathcal{J}}
\newcommand{\calL}{\mathcal{L}}

\newcommand{\scrL}{\mathscr{L}}

\newcommand{\diag}{\mathrm{diag}}
\newcommand{\real}{\mathrm{Re}}
\newcommand{\vecnorm}[1]{\left\| #1 \right\|}
\newcommand{\vecinner}[2]{#1 \cdot #2}
\newcommand{\funcnorm}[2][]{\left| #2 \right|_{#1}}
\newcommand{\funcinner}[3][]{\left\langle #2,#3 \right\rangle_{#1}}
\newcommand{\pd}[2][]{\partial_{#2}^{#1}}

\newtheorem{thm}{Theorem}

\title{\LARGE 
\textbf{Optimal Control for Steady Circulation of a Diffusion Process via Spectral Decomposition of Fokker--Planck Equation}
}

\author{Norihisa Namura and Hiroya Nakao
\thanks{N. Namura is with the Department of Systems and Control Engineering, Institute of Science Tokyo, Tokyo, Japan
        {\tt\small (namura.n.aa@m.titech.ac.jp)}}%
\thanks{H. Nakao is with the Department of Systems and Control Engineering and Research Center for Autonomous Systems Materialogy, Institute of Science Tokyo, Tokyo, Japan
        {\tt\small (nakao@sc.e.titech.ac.jp)}}%
\thanks{N. Namura acknowledges the support from JSPS KAKENHI (No. 25KJ1270).
H. Nakao acknowledges the support from JSPS KAKENHI (Nos. 25H01468, 25K03081, and 22H00516).}
}

\begin{document}

\maketitle
\thispagestyle{empty}
\pagestyle{empty}
\bstctlcite{BSTcontrol}

\begin{abstract}
We present a formulation of an optimal control problem for a two-dimensional diffusion process governed by a Fokker--Planck equation to achieve a nonequilibrium steady state with a desired circulation while accelerating convergence toward the stationary distribution. 
To achieve the control objective, we introduce costs for both the probability density function and flux rotation to the objective functional.
We formulate the optimal control problem through dimensionality reduction of the Fokker--Planck equation via eigenfunction expansion, which requires a low-computational cost.
We demonstrate that the proposed optimal control achieves the desired circulation while accelerating convergence to the stationary distribution through numerical simulations.
\end{abstract}

\section{INTRODUCTION}
\label{sec:introduction}

In the real world, a wide range of systems are affected by stochastic noise. 
Examples include nerve signal transmission by stochastically opening and closing ion channels~\cite{Ermentrout2010mathematical}, 
electronic circuits subjected to thermal fluctuations~\cite{Gardiner2009stochastic}, 
semiclassically approximated quantum dissipative systems~\cite{kato2019semiclassical}, 
and robots subjected to their sensor and actuator noise~\cite{Elamvazhuthi2019mean}.

Random fluctuations often appear as Brownian motion, where particles undergo irregular motion due to collisions with surrounding molecules.
The state dynamics of a Brownian particle is mathematically described by the It\^{o} stochastic differential equation and
the time evolution of the probability density function~(PDF) of the particles can be described by the corresponding Fokker--Planck equation~(FPE)~\cite{Risken1996fokker}.
Formulating control that acts on the FPE is useful~\cite{Annunziato2018fokker}
and numerous studies have formulated optimal control problems for FPEs~\cite{Annunziato2010optimal,Annunziato2013fokker,Fleig2017optimal,Breiten2023improving}.

Optimal control is a general and valuable control approach that has found wide applications, including in the control of oscillatory systems~\cite{Moehlis2006optimal,Namura2024optimal};
however, it is difficult to solve optimal control problems in infinite-dimensional systems.
Since FPE is linear, it can be analyzed using the eigenvalues and eigenfunctions~\cite{Risken1996fokker}.
By using a spectral approach to the FPE, we can approximately formulate optimal control problems, such as for fast convergence to the stationary distribution~\cite{Breiten2017reduction,Breiten2018control,Kalise2025spectral}.

In thermodynamically equilibrium systems, detailed balance ensures that no flux exists in the stationary state, such as the circulation of the particles.
Introducing circulations, however, allows us to enhance efficiency of particle mixing~\cite{Lu2024vortex} or capturing particles in target regions~\cite{Lu2024vortex,Zhang2024vortex}. 
Hence, generating circulations by driving the system in a nonequilibrium state can improve the performance and functional capabilities of the system, rather than merely achieving equilibrium state.

In this study, we formulate an optimal control problem for a two-dimensional diffusion process described by the FPE to achieve two control objectives: 
(i) accelerating the convergence of the PDF toward a stationary state, and 
(ii) generating a desired circulation when the PDF is in the stationary state.
In addition to the cost functional for accelerating convergence to the stationary distribution introduced in~\cite{Kalise2025spectral}, 
a term that accounts for the difference from the desired flux rotation (scalar vorticity~\cite{Flandoli2020fokker}) is added in the objective functional.
In the formulation, we utilize the eigenfunction expansion of the PDF with a finite number of eigenfunctions for dimensionality reduction, 
which is expected to significantly reduce the computational cost during optimization.
Through numerical simulations, we demonstrate that the formulated optimal control problem allows the PDF to converge quickly to its stationary state, while simultaneously generating the desired circulation.

This study is organized as follows.
We first introduce a diffusion process and the corresponding FPE in Sec.~\ref{sec:model}.
In Sec.~\ref{sec:formulation}, we propose an optimal control scheme to accelerate convergence to the stationary distribution while generating a desired circulation based on the spectral dimensionality reduction of the FPE. 
In Sec.~\ref{sec:results}, we demonstrate that the control objective is achieved by the optimal control compared to the uncontrolled case through numerical simulations.
We conclude this study in Sec.~\ref{sec:conclusion}.

\section{FOKKER--PLANCK EQUATION}
\label{sec:model}

We consider a two-dimensional diffusion process~\mbox{$\{ \BX_{t} \}_{t \geq 0}$} on a simply connected compact set \mbox{$\Omega \subset \mathbb{R}^{2}$} with a smooth boundary~$\Gamma$, 
where each state variable is independently subjected to white Gaussian noise.
This diffusion process can be described by the following It\^{o} stochastic differential equation~(SDE):
\begin{align}
\label{eq:sde}
d\BX_{t} = -\nabla V(\BX_{t})dt + \sqrt{2D} d\BW_{t},
\end{align}
where 
\mbox{$t \in \mathbb{R}_{+} = [0,\infty)$} represents the time,
\mbox{$V: \Omega \to \mathbb{R}$} is a smooth potential function,
\mbox{$D > 0$} is a diffusion constant,
\mbox{$\nabla = [\pd{x}~\pd{y}]^{\top}$} represents the gradient for the spatial variable~\mbox{$\Bx = \left[ x~y \right]^{\top} \in \Omega$},
\mbox{$\pd{x} = \partial/\partial x$} represents the partial derivative with respect to~$x$,
\mbox{``$\top$''} represents the transposition of a vector or matrix,
and \mbox{$\{ \BW_{t} \}_{t \geq 0}$} represents a standard two-dimensional Brownian motion.
The time evolution of the probability density function~(PDF)~\mbox{$\rho: \Omega \times \mathbb{R}_{+} \to \mathbb{R}_{+}$} of $\BX_{t}$ is described by the Fokker--Planck equation~(FPE):
\begin{align}
\label{eq:fpe}
\pd{t} \rho(\Bx,t) = \scrL^{*}\rho = \vecinner{\nabla}{(\rho(\Bx,t) \nabla V(\Bx))} + D \Delta \rho(\Bx,t)
\end{align}
by using a linear operator~$\scrL^{*} = \nabla \cdot \nabla V + D \Delta$,
where $\vecinner{\bm{a}}{\bm{b}} = a_{1}b_{1} + a_{2}b_{2}$ is the scalar product of two-dimensional vectors~$\bm{a}$, $\bm{b}$
and $\Delta = \vecinner{\nabla}{\nabla}$ is the Laplacian operator.
We then define a probability flux~$\BJ: \Omega \times \mathbb{R}_{+} \to \mathbb{R}^{2}$ to satisfy
\begin{align}
\scrL^{*} \rho(\Bx,t) = -\vecinner{\nabla}{\BJ(\Bx,t)}.
\end{align}
For the flux $\BJ$, we consider the reflecting boundary condition on the boundary~$\Gamma$ at any~$t$, i.e., 
\begin{align}
\label{eq:no_flux}
\vecinner{\BJ(\Bx,t)}{\Bn} = 0 \quad \mathrm{on} \quad \Gamma \times \mathbb{R}_{+},
\end{align}
where $\Bn$ represents the normal vector perpendicular to $\Gamma$.

In equilibrium systems,
the stationary flux~$\BJ_{\mathrm{s}}$ 
satisfies the detailed balance condition~\mbox{$\BJ_{\mathrm{s}}(\Bx) = \bm{0}$} for any \mbox{$\Bx \in \Omega$}
when the PDF is stationary, \mbox{$\rho(\Bx,t) = \rho_{\mathrm{s}}(\Bx)$}.
In this case, the linear operator~$\scrL^{*}$ is self-adjoint with respect to the weighted square-integrable function space~\mbox{$L^{2}\left(\Omega;\rho_{\mathrm{s}}^{-1}\right)$}~\cite{Pavliotis2014stochastic,Breiten2018control,Kalise2025spectral}.
The stationary distribution is given by 
\begin{align}
\label{eq:stationary}
\rho_{\mathrm{s}}(\Bx) = \frac{1}{Z} \exp\left( - \frac{V(\Bx)}{D} \right),
\end{align}
where $Z = \int_{\Omega} \exp\left( - \frac{V(\Bx)}{D} \right) d\Bx$ is a normalization constant and $d\Bx = dxdy$ is the integration measure.

We can consider the set of the eigenvalues and eigenfunctions~\mbox{$\{ (\lambda_{m}, v_{m}) \},\; (m = 0,1,2,\dots)$} of $\scrL^{*}$ that satisfies
\begin{align}
\scrL^{*} v_{m} = \lambda_{m} v_{m},
\end{align}
where $\{ \lambda_{m} \}$ are sorted as \mbox{$\real~\lambda_{0} \geq \real~\lambda_{1} \geq \cdots$}.
The eigenfunctions~$\{ v_{m} \}$ can be orthonormalized on \mbox{$L^{2}\left(\Omega;\rho_{\mathrm{s}}^{-1}\right)$} as follows:
\begin{align}
\funcinner[L^{2}\left(\Omega;\rho_{\mathrm{s}}^{-1}\right)]{v_{m}}{v_{k}} = \delta_{mk},
\end{align}
where $\funcinner[L^{2}\left(\Omega;\rho_{\mathrm{s}}^{-1}\right)]{p}{q} = \int p(\Bx) \overline{q(\Bx)} \rho_{\mathrm{s}}^{-1}(\Bx) d\Bx$ is the inner product of two functions $p,q$ on \mbox{$L^{2}\left(\Omega;\rho_{\mathrm{s}}^{-1}\right)$},
the overline represents the complex conjugate,
and $\delta_{mk}$ is the Kronecker delta.
Since the stationary distribution satisfies \mbox{$\scrL^{*} \rho_{\mathrm{s}} = 0$}, $\rho_{\mathrm{s}}$ is the eigenfunction that corresponds to the zero eigenvalue,
i.e., \mbox{$\lambda_{0} = 0$} and \mbox{$v_{0} = \rho_{\mathrm{s}}$}.

\section{FORMULATION OF OPTIMAL CONTROL}
\label{sec:formulation}

\subsection{Control Objective}

Our control objective is to accelerate the convergence speed of the PDF toward the stationary state and generate a desired circulation from the initial time~\mbox{$t_{0} \geq 0$} until the final time \mbox{$t_{\mathrm{f}} > t_{0}$}.
For this purpose, we consider a diffusion process with control~$\{ \BX_{t} \}_{t \geq 0}$ described by the following It\^{o} SDE~\eqref{eq:sde_input} whose drift term receives the control inputs~\mbox{$u_{1,2} \in L^{2}([t_{0},t_{\mathrm{f}}])$} with control shape functions~\cite{Breiten2018control}~$\alpha$ and $\phi$:
\begin{align}
\label{eq:sde_input}
d\BX_{t} ={}& \left( - u_{1}(t)\nabla \alpha(\BX_{t}) - u_{2}(t)\frac{1}{\rho_{\mathrm{s}}(\BX_{t})}\nabla^{\bot}\phi(\BX_{t}) \right) dt \cr
&-\nabla V(\BX_{t})dt + \sqrt{2D}dW_{t},
\end{align}
where \mbox{$\nabla^{\bot} = [\pd{y}~-\pd{x}]$} is a partial derivative operator that is perpendicular to $\nabla$.
The term with $u_{1}$ mainly serves to accelerate the convergence, whereas the term with $u_{2}$ generates a circulation in the steady state.
Note that the control inputs~$u_{1,2}$ depend only on time and regulate the intensity of the smooth control shape functions~\mbox{$\alpha: \Omega \to \mathbb{R}$} and \mbox{$\phi: \Omega \to \mathbb{R}$}, respectively.
We assume that the controlled SDE~\eqref{eq:sde_input} admits a unique solution.

The FPE corresponding to the SDE~\eqref{eq:sde_input} is given by 
\begin{align}
\label{eq:fpe_control}
\pd{t} \rho ={}& \vecinner{\nabla}{\left( \rho \left(\nabla V + u_{1}\nabla \alpha + u_{2}\frac{1}{\rho_{\mathrm{s}}}\nabla^{\bot}\phi \right) \right)} + D \Delta \rho
\end{align}
and the probability flux~$\tilde{\BJ}$ under the control is described by 
\begin{align}
\tilde{\BJ} = -\rho\left( \nabla V + u_{1}\nabla \alpha + u_{2}\frac{1}{\rho_{\mathrm{s}}}\nabla^{\bot}\phi \right) - D \nabla \rho.
\end{align}
Here, we assume that the control shape functions satisfy the following boundary conditions:
\begin{align}
\vecinner{\nabla \alpha}{\Bn} &= 0 \quad \mathrm{on} \quad \Gamma, \\
\vecinner{\nabla^{\bot} \phi}{\Bn} &= 0 \quad \mathrm{on} \quad \Gamma,
\end{align}
in order that $\tilde{\BJ}$ satisfies the reflecting boundary condition even under the control, i.e., 
\begin{align}
\label{eq:no_flux_control}
\vecinner{\tilde{\BJ}(\Bx,t)}{\Bn} = 0 \quad \mathrm{on} \quad \Gamma \times \mathbb{R}_{+}.
\end{align}
We assume the well-posedness of the controlled FPE~\eqref{eq:fpe_control} with the boundary condition~\eqref{eq:no_flux_control}.

The control objective includes fast convergence to the stationary distribution~$\rho_{\mathrm{s}}$ of the uncontrolled system.
Here, we consider that the system to have converged to the stationary state when $\rho$ is sufficiently close to $\rho_{\mathrm{s}}$ within a finite time.
The stationary flux~$\tilde{\BJ}_{\mathrm{s}}$ in the stationary state with \mbox{$\rho = \rho_{\mathrm{s}}$} should satisfy the following condition for all \mbox{$\Bx \in \Omega$}, i.e., 
\begin{align}
\label{eq:stationary_condition}
\vecinner{\nabla}{\tilde{\BJ}_{\mathrm{s}}(\Bx)} = 0.
\end{align}
For this condition to hold, it is sufficient that the following two equations are satisfied:
\begin{align}
\label{eq:alpha_condition}
u_{1} \vecinner{\nabla}{\left( \rho_{\mathrm{s}}\nabla \alpha \right)} &= 0, \\
\label{eq:phi_condition}
u_{2} \vecinner{\nabla}{\left( \nabla^{\bot} \phi \right)} &= 0.
\end{align}
In general, \mbox{$\vecinner{\nabla}{(\rho_{\mathrm{s}}\nabla \alpha)}$} does not vanish
but it is sufficient that the condition~\eqref{eq:alpha_condition} is satisfied when \mbox{$u_{1} = 0$},
because $u_{1}$ becomes practically zero after the sufficient convergence of $\rho$ to $\rho_{\mathrm{s}}$.
Meanwhile, the condition~\eqref{eq:phi_condition} is always satisfied because $\vecinner{\nabla}{\left( \nabla^{\bot} \phi \right)} = 0$ holds for any smooth function~$\phi$:
%
%
Therefore, when \mbox{$u_{1} = 0$}, $\rho_{\mathrm{s}}$ is still a stationary distribution of the FPE~\eqref{eq:fpe_control} even for \mbox{$u_{2} \neq 0$}.

To generate a desired circulation, we control the system so that the flux rotation~\mbox{$\omega(\Bx,t) = \nabla \times \tilde{\BJ}(\Bx,t)$} reaches a desired flux rotation~$\omega_{\mathrm{d}}(\Bx)$ in the steady state with $\rho = \rho_{\mathrm{s}}$. 
The term~\mbox{$\nabla \times \tilde{\BJ}(\Bx,t)$} yields a scalar quantity in two dimensions and satisfies \mbox{$\nabla \times \tilde{\BJ}(\Bx,t) = -\vecinner{\nabla^{\bot}}{\tilde{\BJ}(\Bx,t)}$}.
Here, the flux rotation satisfies \mbox{$\int_{\Omega} \omega(\Bx,t) d\Bx = 0$} at any~$t$ and we further assume that \mbox{$\omega(\cdot,t) \in L^{2}\left(\Omega;\rho_{\mathrm{s}}^{-1}\right)$}.
Since we can assume \mbox{$u_{1} = 0$} when \mbox{$\rho = \rho_{\mathrm{s}}$}, 
the stationary flux for a constant $u_{2}$ satisfies \mbox{$\tilde{\BJ}_{\mathrm{s}} = -u_{2}\nabla^{\bot} \phi$} and 
the stationary flux rotation satisfies $\omega_{\mathrm{s}} = \nabla \times \tilde{\BJ}_{\mathrm{s}} = u_{2} \vecinner{\nabla^{\bot}}{\nabla^{\bot} \phi} = u_{2} \Delta \phi$, whose shape depends only on $\phi$.
Thus, we can achieve any smooth desired flux rotation~\mbox{$\omega_{\mathrm{d}} = \Delta \phi$} by appropriately designing the shape function~$\phi$.
In this case, the control objective can be set as $u_{2} = 1$.

In this study, we assume the controllability of the system to the nonequilibrium steady state
characterized by the stationary distribution~$\rho_{\mathrm{s}}$ and the desired flux rotation~$\omega_{\mathrm{d}}$.
We note that the detailed balance condition does not hold in the steady state under the control,
but we can discuss the controlled FPE~\eqref{eq:fpe_control} using the fact that the detailed balance condition~\mbox{$\BJ_{\mathrm{s}} = \bm{0}$} holds in the original uncontrolled system.

\subsection{Spectral Dimensionality Reduction}
We consider the eigenfunction expansion of the PDF~$\rho$ using the eigenfunctions $\{ v_{m} \}$ of the linear operator~$\scrL^{*}$ as
\begin{align}
\label{eq:eigenfunction_expansion}
\rho(\Bx,t) = \sum_{m=0}^{\infty} c_{m}(t) v_{m}(\Bx),
\end{align}
where each coefficient~$c_{m}$ can be obtained by 
\begin{align}
\label{eq:coefficient}
c_{m}(t) = \funcinner[L^{2}\left(\Omega;\rho_{\mathrm{s}}^{-1}\right)]{\rho(\cdot,t)}{v_{m}}.
\end{align}
By defining $\zeta_{1,2}$ as 
\begin{align}
\zeta_{1}(\Bx,t) &= \vecinner{\nabla}{(\rho(\Bx,t) \nabla \alpha(\Bx))}, \\
\zeta_{2}(\Bx,t) &= \vecinner{\nabla}{\left( \rho(\Bx,t)\frac{1}{\rho_{\mathrm{s}}(\Bx)}\nabla^{\bot}\phi(\Bx) \right)},
\end{align}
the controlled FPE~\eqref{eq:fpe_control} can be rewritten as 
\begin{align}
\pd{t} \rho ={}& \scrL^{*} \rho + u_{1}\zeta_{1} + u_{2}\zeta_{2} \cr
\coloneqq{}& \calF\{ \rho,u_{1},u_{2} \},
\end{align}
where $\calF$ describes the dynamics of the controlled FPE.
The dynamics of each coefficient~$c_{m}$ can then be described by the following ordinary differential equation~(ODE):
\begin{align}
\dot{c}_{m} ={}& \lambda_{m}c_{m} + u_{1} \funcinner[L^{2}\left(\Omega;\rho_{\mathrm{s}}^{-1}\right)]{\zeta_{1}(\cdot,t)}{v_{m}} \cr
&+ u_{2} \funcinner[L^{2}\left(\Omega;\rho_{\mathrm{s}}^{-1}\right)]{\zeta_{2}(\cdot,t)}{v_{m}} \cr
={}& \lambda_{m}c_{m} + u_{1} \sum_{k=0}^{\infty} c_{k} \funcinner[L^{2}\left(\Omega;\rho_{\mathrm{s}}^{-1}\right)]{\vecinner{\nabla}{(v_{k} \nabla \alpha)}}{v_{m}} \cr
&+ u_{2} \sum_{k=0}^{\infty} c_{k} \funcinner[L^{2}\left(\Omega;\rho_{\mathrm{s}}^{-1}\right)]{\vecinner{\nabla}{\left( \frac{v_{k}}{\rho_{\mathrm{s}}} \nabla^{\bot}\phi \right)}}{v_{m}},~~~
\end{align}
where $\dot{c}_{m}$ represents the time derivative of $c_{m}$.

Hereafter, we consider a finite set of $M$~eigenvalues~$\{ \lambda_{m} \}_{m=0}^{M-1}$ and corresponding eigenfunctions~$\{ v_{m} \}_{m=0}^{M-1}$
for the purpose of dimensionality reduction.
We define a coefficient vector and an eigenvalue matrix by
\begin{align}
\Bc &= 
\begin{bmatrix}
c_{0} & c_{1} & \cdots & c_{M-1}
\end{bmatrix}
^{\top} \in \mathbb{R}^{M}, \\
\BLambda &= \mathrm{diag}(\lambda_{0},\lambda_{1},\dots,\lambda_{M-1}) \in \mathbb{R}^{M \times M}, 
\end{align}
respectively, and introduce matrices~\mbox{$\BB_{1,2} \in \mathbb{R}^{M \times M}$},
where the $(m,k)$-components~$B_{1,2}^{(mk)}$ of the matrices are given by
\begin{align}
B_{1}^{(mk)} &= \funcinner[L^{2}\left(\Omega;\rho_{\mathrm{s}}^{-1}\right)]{\vecinner{\nabla}{(v_{k} \nabla \alpha)}}{v_{m}}, \\
B_{2}^{(mk)} &= \funcinner[L^{2}\left(\Omega;\rho_{\mathrm{s}}^{-1}\right)]{\vecinner{\nabla}{\left( \frac{v_{k}}{\rho_{\mathrm{s}}} \nabla^{\bot} \phi \right)}\!}{v_{m}},
\end{align}
respectively.
Finally, the controlled FPE~\eqref{eq:fpe_control} can be approximated by the following ODE:
\begin{align}
\label{eq:state_equation}
\dot{\Bc} ={}& \BLambda\Bc + u_{1}\BB_{1}\Bc + u_{2}\BB_{2}\Bc \cr
\coloneqq{}& \Bf(\Bc,u_{1},u_{2}).
\end{align}
The ODE~\eqref{eq:state_equation} driven by the vector field~$\Bf$ is a bilinear system with the state~$\Bc$ and control inputs~$u_{1,2}$.

\subsection{Optimization}

To achieve the control objective, we introduce a stage cost functional~$\calL_{1}$ for the PDF, which is described using a weight~\mbox{$Q_{1} > 0$} as follows:
\begin{align}
\calL_{1}[\rho] ={}& \frac{Q_{1}}{2} \funcnorm[L^{2}\left(\Omega;\rho_{\mathrm{s}}^{-1}\right)]{\rho(\cdot,t) - \rho_{\mathrm{s}}}^{2},
\end{align}
where $\funcnorm[L^{2}\left(\Omega;\rho_{\mathrm{s}}^{-1}\right)]{\cdot}$ represents the \mbox{$L^{2}\left(\Omega;\rho_{\mathrm{s}}^{-1}\right)$} norm.
In addition, we introduce a stage cost functional~$\calL_{2}$ for the flux rotation, which is described using a weight~\mbox{$Q_{2} > 0$} as follows:
\begin{align}
\calL_{2}[\omega] ={}& \frac{Q_{2}}{2} \funcnorm[L^{2}\left(\Omega;\rho_{\mathrm{s}}^{-1}\right)]{\omega(\cdot,t) - \omega_{\mathrm{d}}}^{2} \cr 
={}& \frac{Q_{2}}{2} \funcnorm[L^{2}\left(\Omega;\rho_{\mathrm{s}}^{-1}\right)]{\nabla \times \tilde{\BJ}(\cdot,t) - \omega_{\mathrm{d}}}^{2}.
\end{align}
We also introduce a terminal cost functional~$\varphi$ for the PDF and flux rotation using a similar formulation as
\begin{align}
\varphi[\rho(\cdot,t_{\mathrm{f}}),\omega(\cdot,t_{\mathrm{f}})] ={}& \frac{Q_{\mathrm{f}}}{2} \funcnorm[L^{2}\left(\Omega;\rho_{\mathrm{s}}^{-1}\right)]{\rho(\cdot,t_{\mathrm{f}}) - \rho_{\mathrm{s}}}^{2} \cr
&+ \frac{R_{\mathrm{f}}}{2} \funcnorm[L^{2}\left(\Omega;\rho_{\mathrm{s}}^{-1}\right)]{\omega(\cdot,t_{\mathrm{f}}) - \omega_{\mathrm{d}}}^{2}.~~~
\end{align}
Finally, we introduce a stage cost function~$\calL_{3}$ for the inputs~$u_{1,2}$ using weights~$R_{1,2} > 0$ as follows:
\begin{align}
\calL_{3}(u_{1}(t),u_{2}(t)) = \frac{1}{2} R_{1}u_{1}(t)^{2} + \frac{1}{2} R_{2}u_{2}(t)^{2}.
\end{align}
Hence, the objective functional $\calJ$ is represented as 
\begin{align}
\calJ[u_{1},u_{2}] ={}& \varphi[\rho(\cdot,t_{\mathrm{f}}),\omega(\cdot,t_{\mathrm{f}})] \cr
&+ \int_{t_{0}}^{t_{\mathrm{f}}} \calL_{1}[\rho] + \calL_{2}[\omega] + \calL_{3}(u_{1},u_{2}) dt.~~
\end{align}

The optimal control problem is formulated as 
\begin{align}
\label{eq:optimal_control}
\min_{u_{1},u_{2}} & \quad \calJ[u_{1},u_{2}] \cr
\mathrm{s.t.} & \quad \pd[]{t} \rho = \calF\{ \rho, u_{1}, u_{2} \}, \cr
& \quad \rho(t_{0}) = \rho_{0},
\end{align}
where $\rho_{0}$ is the initial PDF at $t_{0}$.
The proposed optimal control problem satisfies the standard conditions ensuring the existence of an optimal control, namely coercivity and weak lower semicontinuity of the cost functional together with continuity of the control-to-state mapping~\cite{Lions1971optimal,Kalise2025spectral}.

Next, we reformulate the objective functional by using the eigenfunction expansion~\eqref{eq:eigenfunction_expansion} of the PDF~$\rho$.
The stage cost for the PDF can be approximated by a function~$\tilde{\calL}_{1}$ of the coefficient vector~$\Bc$ as 
\begin{align}
& \calL_{1}[\rho] \cr
\approx{}&\tilde{\calL}_{1}(\Bc) \cr
={}& \frac{Q_{1}}{2} \funcinner[L^{2}\left(\Omega;\rho_{\mathrm{s}}^{-1}\right)]{\sum_{m=0}^{M-1} \left(c_{m} - c_{\mathrm{s}}^{(m)}\right)v_{m}}{\sum_{k=0}^{M-1} \left(c_{k} - c_{\mathrm{s}}^{(k)}\right)v_{k}} \cr
={}& \frac{Q_{1}}{2} \vecnorm{\Bc - \Bc_{\mathrm{s}}}^{2},
\end{align}
where the $m$-th component of the vector~$\Bc_{\mathrm{s}} \in \mathbb{R}^{M}$ is obtained as \mbox{$c_{\mathrm{s}}^{(m)} = \funcinner[L^{2}\left(\Omega;\rho_{\mathrm{s}}^{-1}\right)]{\rho_{\mathrm{s}}}{v_{m}}$} and 
$\vecnorm{\cdot}$ represents the Euclid norm of a $M$-dimensional vector.
The stage cost for the flux rotation can also be approximated by a function~$\tilde{\calL}_{2}$ of $\Bc$ and $u_{1,2}$ by calculating the flux~$\tilde{\BJ}$ using the eigenfunction expansion~\eqref{eq:eigenfunction_expansion}, which can be written as 
\begin{align}
\calL_{2}[\omega] 
\approx{}&\tilde{\calL}_{2}(\Bc,u_{1},u_{2}) \cr
={}& \frac{Q_{2}}{2} \vecnorm{(\BA_{1} + u_{1}(t)\BA_{2} + u_{2}(t)\BA_{3})\Bc - \Bd}^{2}.~
\end{align}
Here, each component of the matrices~\mbox{$\BA_{1,2,3} \in \mathbb{R}^{M \times M}$} and the vector~\mbox{$\Bd \in \mathbb{R}^{M}$} is represented as
\begin{align}
A_{1}^{(mk)} &= \funcinner[L^{2}\left(\Omega;\rho_{\mathrm{s}}^{-1}\right)]{\vecinner{\nabla^{\bot}v_{m}}{\nabla V}}{v_{k}}, \\
A_{2}^{(mk)} &= \funcinner[L^{2}\left(\Omega;\rho_{\mathrm{s}}^{-1}\right)]{\vecinner{\nabla^{\bot}v_{m}}{\nabla \alpha}}{v_{k}}, \\
A_{3}^{(mk)} &= \funcinner[L^{2}\left(\Omega;\rho_{\mathrm{s}}^{-1}\right)]{\vecinner{\nabla^{\bot}\frac{v_{m}}{\rho_{\mathrm{s}}}}{\nabla^{\bot} \phi + \frac{v_{m}}{\rho_{\mathrm{s}}}\Delta\phi}}{v_{k}}, \\
d^{(k)} &= \funcinner[L^{2}\left(\Omega;\rho_{\mathrm{s}}^{-1}\right)]{\omega_{\mathrm{d}}}{v_{k}},
\end{align}
respectively, because the flux~$\tilde{\BJ}$ can be approximated as 
\begin{align}
\label{eq:flux_expansion}
\tilde{\BJ} \approx{}& - \sum_{m=0}^{M-1} c_{m}v_{m} \left( \nabla V + u_{1}\nabla \alpha + u_{2}\frac{1}{\rho_{\mathrm{s}}}\nabla^{\bot}\phi \right) \cr
&- D \sum_{m=0}^{M-1} c_{m} \nabla v_{m}.
\end{align}
Since \mbox{$u_{2} = 1$} is the control objective for the flux rotation in the stationary state, the terminal cost can be approximated as 
\begin{align}
&\varphi[\rho(\cdot,t_{\mathrm{f}}),\omega(\cdot,t_{\mathrm{f}})] \cr
\approx{}& \tilde{\varphi}(\Bc(t_{\mathrm{f}}),u_{2}(t_{\mathrm{f}})) \cr
={}& \frac{Q_{\mathrm{f}}}{2} \vecnorm{\Bc(t_{\mathrm{f}}) - \Bc_{\mathrm{s}}}^{2} + \frac{R_{\mathrm{f}}}{2} (u_{2}(t_{\mathrm{f}}) - 1)^{2}.
\end{align}
Hence, the objective functional can be approximated as 
\begin{align}
&\calJ[u_{1},u_{2}] \cr
\approx{}& \tilde{\calJ}[u_{1},u_{2}] \cr
={}& \tilde{\varphi}(\Bc(t_{\mathrm{f}}),u_{2}(t_{\mathrm{f}})) \cr
&+ \int_{t_{0}}^{t_{\mathrm{f}}} \tilde{\calL}_{1}(\Bc) + \tilde{\calL}_{2}(\Bc,u_{1},u_{2}) + \calL_{3}(u_{1},u_{2}) dt.
\end{align}
The objective functional includes the product of $\Bc$ and $u_{1,2}$ but is still coercive.

Thus, the low-dimensional optimal control problem can be formulated with a Lagrange multiplier vector \mbox{$\Bmu(t) \in \mathbb{R}^{M}$} as 
\begin{align}
\label{eq:optimal_control_approx}
&\min_{u_{1},u_{2}} \quad \tilde{\calJ}[u_{1},u_{2}] + \int_{t_{0}}^{t_{\mathrm{f}}} \Bmu(t)^{\top} \left( \Bf(\Bc,u_{1},u_{2}) - \dot{\Bc}(t) \right) dt \cr
&\mathrm{s.t.} \quad \quad \Bc(t_{0}) = \Bc_{0},
\end{align}
given that the initial condition at $t_{0}$ for the state~$\Bc$ satisfies
\begin{align}
\label{eq:opt_initial}
\Bc(t_{0}) = \Bc_{0},
\end{align}
where $\Bc_{0}$ is the initial state corresponding to $\rho_{0}$ and 
its $m$-th component is obtained as \mbox{$(c_{0})^{(m)} = \funcinner[L^{2}\left(\Omega;\rho_{\mathrm{s}}^{-1}\right)]{\rho_{0}}{v_{m}}$}.
\begin{thm}
In the $M \to \infty$ limit, any optimal solution of the approximated problem~\eqref{eq:optimal_control_approx} converges to an optimal solution of the original problem~\eqref{eq:optimal_control}.
\end{thm}
\begin{proof}
Since the eigenfunction expansion of $\rho$ up to $m = M-1$ converges to $\rho$ as $M \to \infty$ in $L^{2}\left(\Omega;\rho_{\mathrm{s}}^{-1}\right)$ space~\cite{Pavliotis2014stochastic},
the cost function~$\tilde{\calL}_{1}$ converges to the cost functional~$\calL_{1}$.
Since the flux~$\tilde{\BJ}$ can be expanded as Eq.~\eqref{eq:flux_expansion} and \mbox{$\omega(\Bx,t) = -\vecinner{\nabla^{\bot}}{\tilde{\BJ}(\Bx,t)}$},
the cost function~$\tilde{\calL}_{2}$ converges to the cost functional~$\calL_{2}$.
Similarly, the cost function~$\tilde{\varphi}$ converges to the cost functional~$\varphi$.
Therefore, the value of $\tilde{\calJ}$ converges to that of $\calJ$.
Since $\tilde{\calJ}$ is coercive independently of $M$, any optimal solution of \eqref{eq:optimal_control_approx} remains bounded for all $M$ and converges to an optimal solution of \eqref{eq:optimal_control}.
\end{proof}

Defining the Hamiltonian~$H(\Bc,u_{1},u_{2},\Bmu)$ by
\begin{align}
H(\Bc,u_{1},u_{2},\Bmu) ={}& \tilde{\calL}_{1}(\Bc) + \tilde{\calL}_{2}(\Bc,u_{1},u_{2}) + \calL_{3}(u_{1},u_{2}) \cr
&+ \Bmu(t)^{\top}\Bf(\Bc,u_{1},u_{2}),
\end{align}
the dynamics of $\Bmu$ can be described by the following ODE:
\begin{align}
\label{eq:opt_adjoint}
\dot{\Bmu} &= -\nabla_{\Bc} H(\Bc,u_{1},u_{2},\Bmu)
\end{align}
and $\Bmu$ satisfies the terminal condition:
\begin{align}
\label{eq:opt_final}
\Bmu(t_{\mathrm{f}}) = \nabla_{\Bc} H(\Bc(t_{\mathrm{f}}),u_{2}(t_{\mathrm{f}})),
\end{align}
where $\nabla_{\Bc}$ is the gradient with respect to $\Bc$.

Since it is difficult to obtain the optimal solution analytically, we numerically obtained it using the MATLAB optimization solver~\texttt{fminunc}.
At each step of the optimization, the state equation~\eqref{eq:state_equation} is first numerically integrated from \mbox{$t = t_{0}$} to \mbox{$t = t_{\mathrm{f}}$} using the Euler method with the initial condition~\eqref{eq:opt_initial},
and then the ODE for $\Bmu$ given by~\eqref{eq:opt_adjoint} is solved from \mbox{$t = t_{\mathrm{f}}$} to \mbox{$t = t_{0}$} with the terminal condition~\eqref{eq:opt_final},
to evaluate the value of the objective functional~$\tilde{\calJ}$.
Finally, we obtained a local optimal solution that satisfies the stationarity condition of the Hamiltonian with respect to the inputs, i.e., $\partial H / \partial u_{1,2} = 0$ for all \mbox{$t \in [t_{0},t_{\mathrm{f}}]$}.
Although this optimization problem is not convex, the formulated optimization problem requires low computational cost owing to the dimensionality reduction of the original problem.

\section{RESULTS}
\label{sec:results}

We demonstrate through numerical simulations that the FPE~\eqref{eq:fpe_control} under the optimal control obtained by the proposed formulation sufficiently achieves the control objective.
As an example, we consider a controlled diffusion process~\eqref{eq:sde_input} with the potential~\mbox{$V(\Bx) = 2x^{2} + 3y^{2}$} and diffusion constant~\mbox{$D = 2$}.
Without control, the FPE~\eqref{eq:fpe} has a stationary PDF~$\rho_{\mathrm{s}}$ as shown in Fig.~\ref{fig:stationary}
on a domain~\mbox{$\Omega = [-L,L] \times [-L,L]$} with a length~\mbox{$L = 4$}.
For numerical simulations, the $x$ and $y$ components are discretized with intervals~\mbox{$\Delta_{x} = 0.08$} and \mbox{$\Delta_{y} = 0.08$}, respectively.

The control objective is to accelerate convergence toward the stationary distribution~$\rho_{\mathrm{s}}$ and generate a desired flux rotation~$\omega_{\mathrm{d}}$ as shown in Fig.~\ref{fig:rotation}
from the initial time~\mbox{$t_{0} = 0$} until the final time \mbox{$t_{\mathrm{f}} = 1$}.
Here, the control shape functions are set as 
\begin{align}
\alpha(\Bx) = \cos\left( \frac{\pi x}{2L} \right) \cos\left( \frac{\pi y}{2L} \right) 
\end{align}
and 
\begin{align}
\phi(\Bx) = \rho_{\mathrm{s}} \frac{L^2}{4} \left( \frac{x^{2}}{L^{2}} - 1 \right) \left( \frac{y^{2}}{L^{2}} - 1 \right),
\end{align}
which satisfy the boundary conditions.
The initial distribution is taken as a mixture of two truncated and normalized Gaussian PDFs with equal weights;
one has a mean~$[1~0]^{\top}$ and covariance matrix~$\diag(0.5,0.5)$ and the other has a mean~$[-1~0]^{\top}$ and the same covariance.
We approximate the PDF driven by the controlled FPE~\eqref{eq:fpe_control} using \mbox{$M = 21$} eigenfunctions corresponding to the \mbox{$M = 21$~largest real eigenvalues}.

We performed optimization with the weights~$Q_{\mathrm{f}} = 10^{2}$, $R_{\mathrm{f}} = 1$, $Q_{1} = 10^{4}$, $Q_{2} = 10$, $R_{1} = 1$, and $R_{2} = 1$
from an initial guess~\mbox{$u_{1}(t) = 0$} and \mbox{$u_{2}(t) = 1$}.
In the optimization, we numerically integrate the ODEs~\eqref{eq:state_equation} and \eqref{eq:opt_adjoint} by the Euler method with a time step~\mbox{$\Delta_{t} = 5 \times 10^{-3}$}.
We show the optimal control inputs in Figs.~\ref{fig:input}(a) and (b).
The input~$u_{1}$ increases from a large negative value to $0$ in the initial stage, indicating that it is optimal to apply a large input in the early state to accelerate the convergence of $\rho$ towards $\rho_{\mathrm{s}}$.
Meanwhile, the input~$u_{2}$ increases from $0$ and converges to $1$ after the initial stage, indicating that it is optimal to generate the desired circulation when the PDF sufficiently convergences to the target $\rho_{\mathrm{s}}$.

We numerically integrated the FPE~\eqref{eq:fpe_control} using the obtained optimal control inputs by the implicit method with the time step~\mbox{$\Delta_{t}$}.
We compare the optimal control and uncontrolled cases by calculating the \mbox{$L^{2}\left(\Omega;\rho_{\mathrm{s}}^{-1}\right)$}~norm~\mbox{$e_{\rho} = \funcnorm[L^{2}\left(\Omega;\rho_{\mathrm{s}}^{-1}\right)]{\rho(\cdot,t) - \rho_{\mathrm{s}}}$} for the PDF and 
\mbox{$L^{2}\left(\Omega;\rho_{\mathrm{s}}^{-1}\right)$}~norm~\mbox{$e_{\omega} = \funcnorm[L^{2}\left(\Omega;\rho_{\mathrm{s}}^{-1}\right)]{\omega(\cdot,t) - \omega_{\mathrm{d}}}$} for the flux rotation.
We can find that the control objective is successfully achieved by using the optimal control inputs.
We note that the initial values of the $e_{\omega}$ are different between the optimal control and uncontrolled cases because the flux rotation is affected by the control inputs.

Furthermore, we numerically integrated the time evolution of $10^{5}$ independent particles governed by the controlled SDE~\eqref{eq:sde_input} using the Euler--Maruyama method with the time step $\Delta_{t}$.
We then computed the flux at the final time, which can be obtained from the particle states
by calculating the drift term by the sample mean of the displacement at each grid domain and approximating the histogram by a truncated Gaussian mixture model.
We show the flux under optimal control at the final time in Fig.~\ref{fig:velocity}, 
where the flux successfully exhibits clockwise rotation around the origin, corresponding to $\omega_{\mathrm{d}}$ being negative near the origin.

\begin{figure}
\centering
\includegraphics[width=0.34\textwidth]{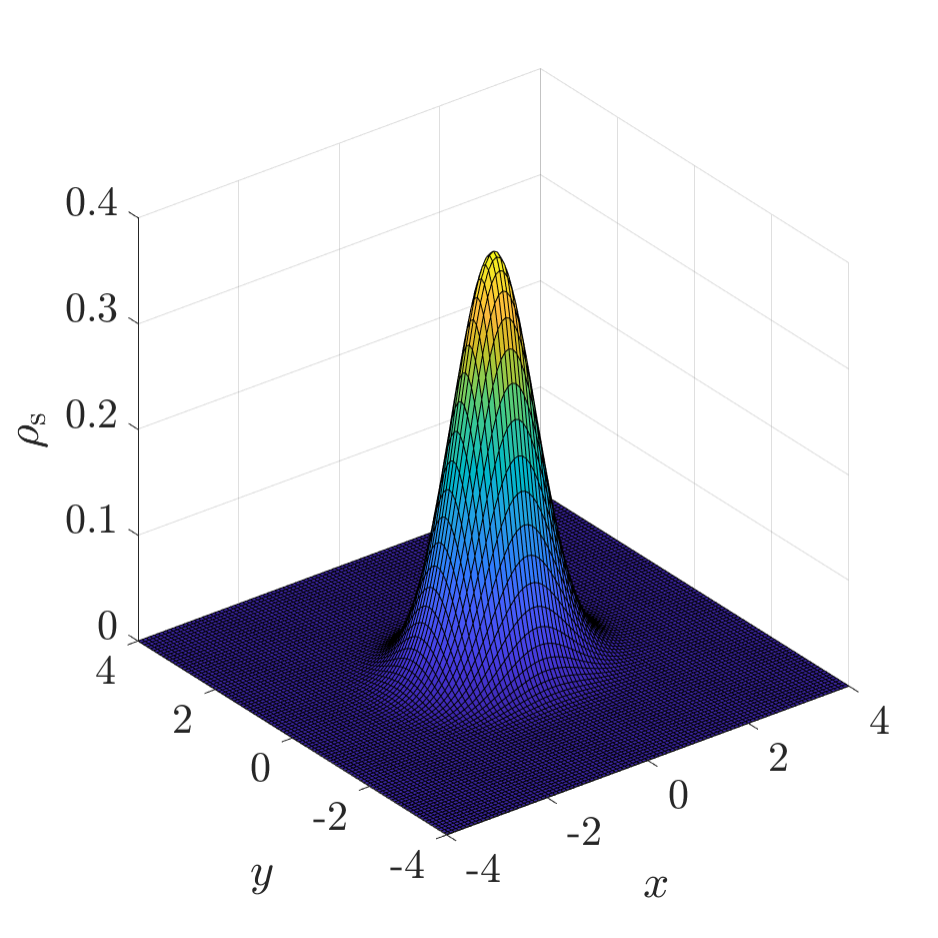}
\caption{Stationary distribution~$\rho_{\mathrm{s}}$ of the uncontrolled FPE.}
\label{fig:stationary}
\end{figure}

\begin{figure}
\centering
\includegraphics[width=0.34\textwidth]{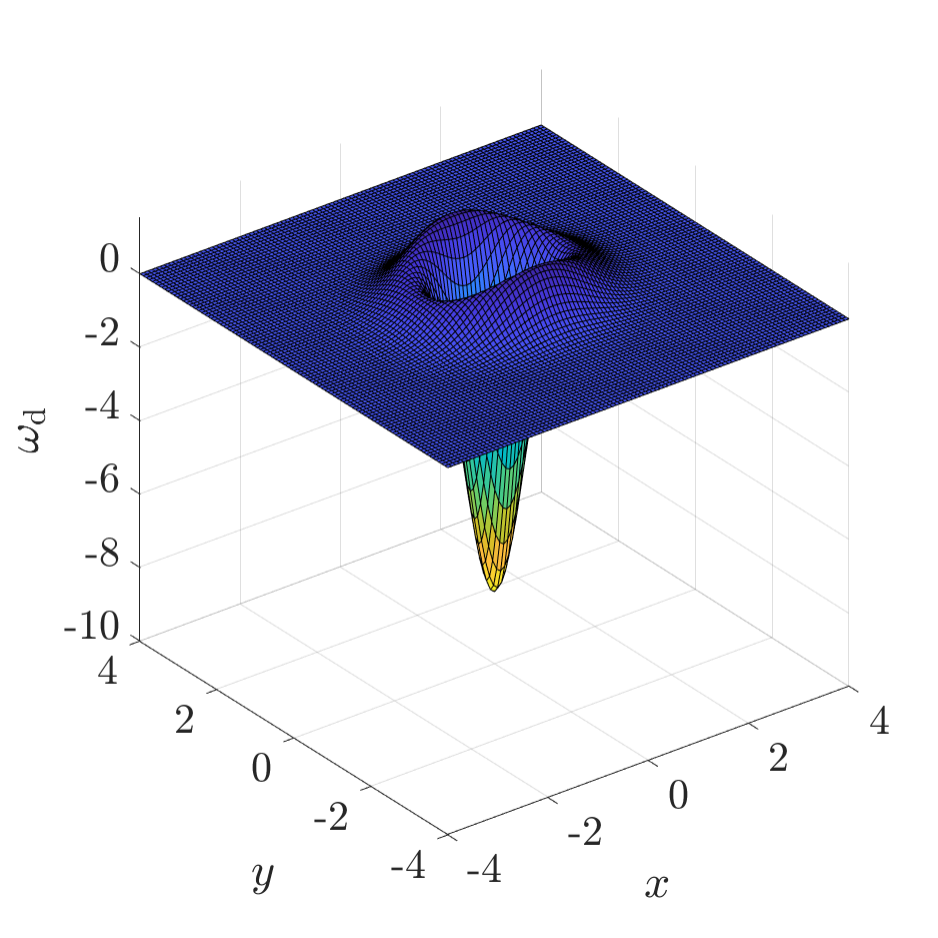}
\caption{Desired flux rotation~$\omega_{\mathrm{d}}$.}
\label{fig:rotation}
\end{figure}

\begin{figure}
\centering
\includegraphics[width=0.48\textwidth]{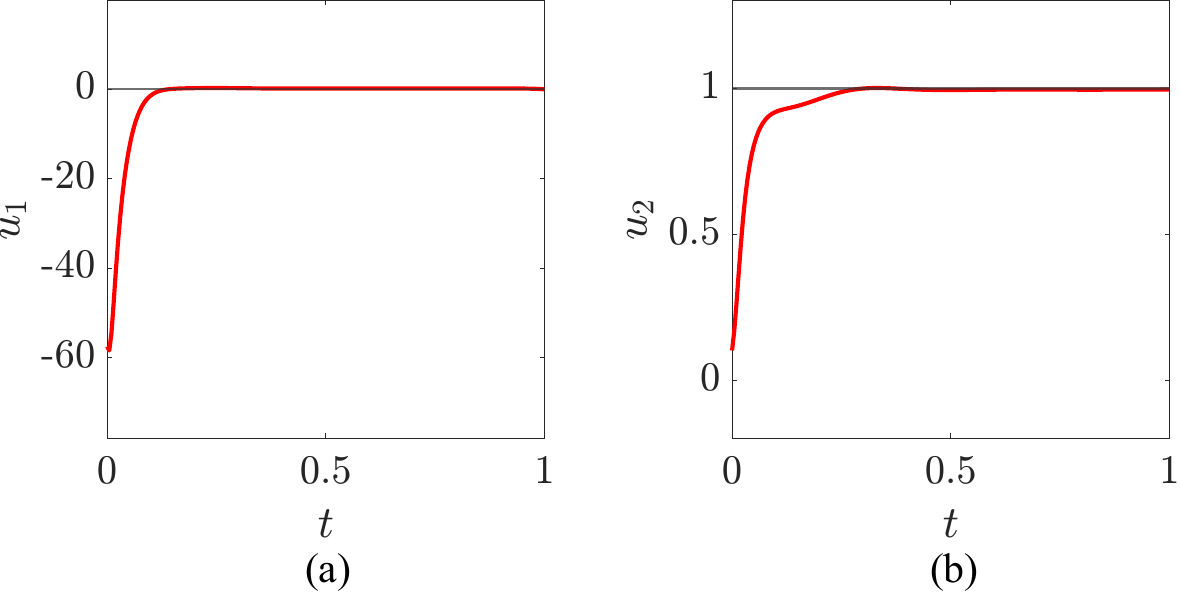}
\caption{
Optimal control inputs.
(a)~Optimal control input~$u_{1}$.
The horizontal line represents $u_{1}(t) = 0$.
(b)~Optimal control input~$u_{2}$.
The horizontal line represents $u_{2}(t) = 1$.
}
\label{fig:input}
\end{figure}

\begin{figure}
\centering
\includegraphics[width=0.48\textwidth]{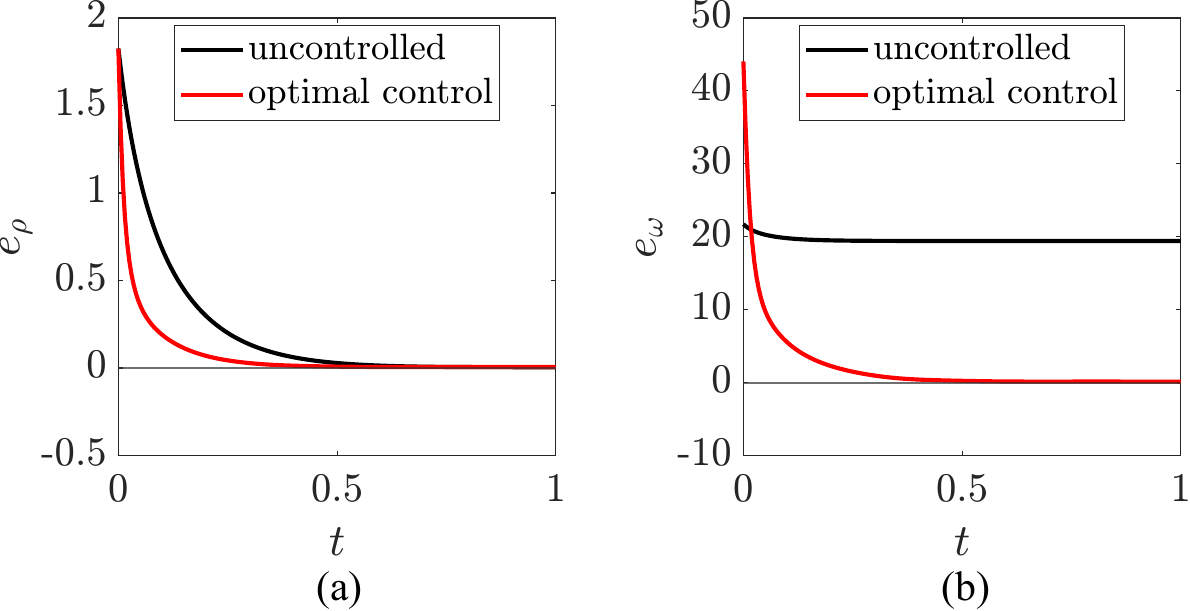}
\caption{
Comparison of the \mbox{$L^{2}\left(\Omega;\rho_{\mathrm{s}}^{-1}\right)$}~norm between the uncontrolled case and optimal control case.
(a)~\mbox{$L^{2}\left(\Omega;\rho_{\mathrm{s}}^{-1}\right)$}~norm between the stationary and current distribution.
(b)~\mbox{$L^{2}\left(\Omega;\rho_{\mathrm{s}}^{-1}\right)$}~norm between the desired and current flux rotation.
}
\label{fig:L2_norm}
\end{figure}

\begin{figure}
\centering
\includegraphics[width=0.48\textwidth]{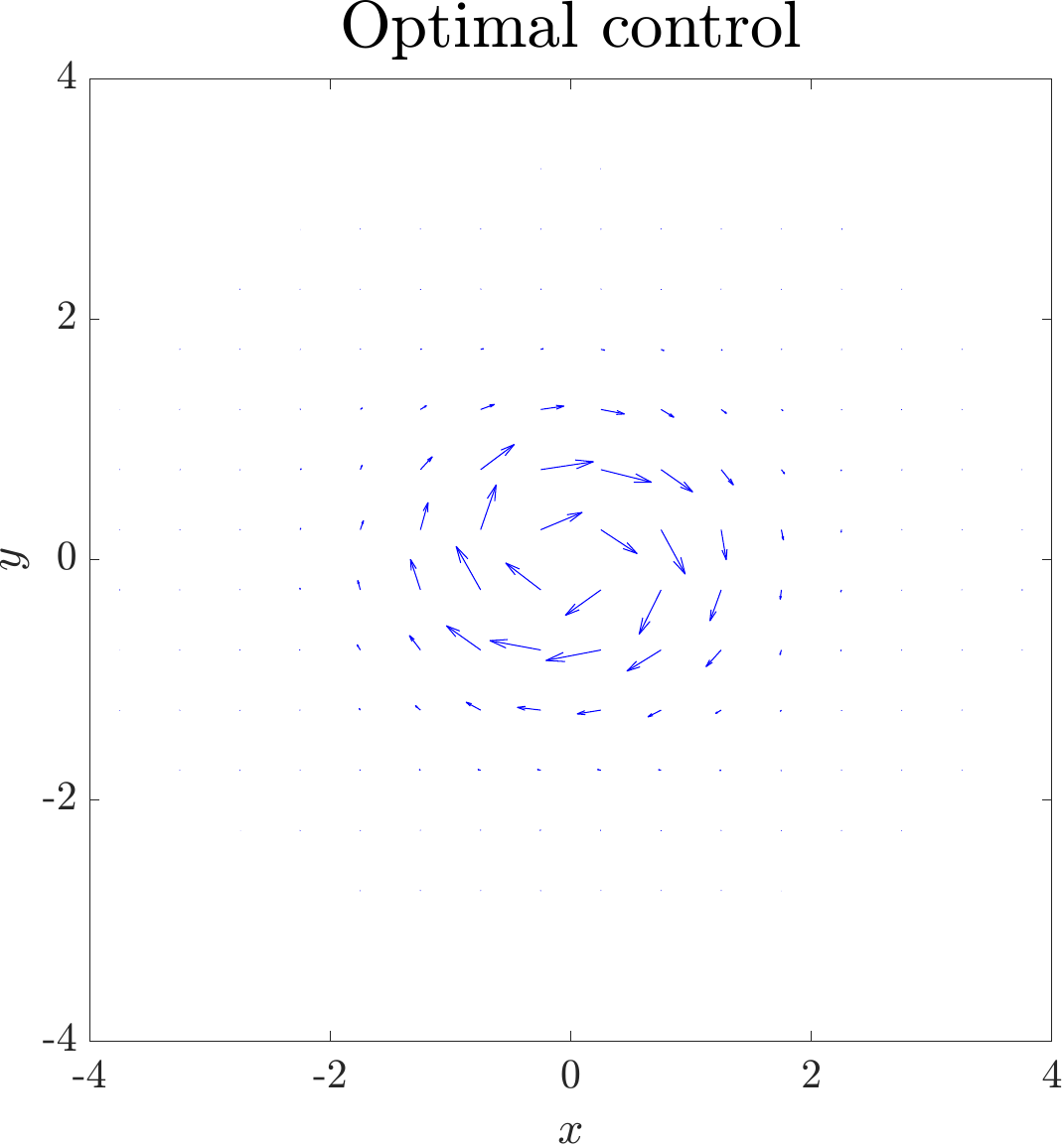}
\caption{
The behavior of the flux obtained from the particles under the optimal control at the final time on the $xy$ plane.
The length of the arrows represents the flux speed.
}
\label{fig:velocity}
\end{figure}

\section{CONCLUSIONS}

\label{sec:conclusion}

In this study, we formulated an optimal control problem for a diffusion process described by a two-dimensional It\^{o} SDE to achieve a desired circulation while accelerating convergence to the stationary distribution of the original circulation-free Fokker--Planck equation.
By the formulation of the optimal control problem based on dimensionality reduction via the eigenfunction expansion of the probability density function,
we demonstrated that the control objective can be achieved through numerical simulations.
Although the desired flux rotation can also be achieved while ensuring convergence to the stationary distribution by the trivial choice of the control inputs, \mbox{$u_{1}(t) = 0$} and \mbox{$u_{2}(t) = 1$}, 
it is worthwhile to solve the optimal control problem 
because a different and more efficient optimal solution (Fig.~\ref{fig:input}) was obtained when the trivial choice was used as the initial guess in the numerical example.
A future work would be to formulate optimal control problems for the similar control objective in three or higher-dimensional cases and 
investigate a nonlinear FPE under particle interactions.


\bibliographystyle{ieeetran}
\bibliography{references}

\end{document}